\documentclass[aps, prd, preprint, groupedaddress,showpacs,18pt]{revtex4}
\usepackage{hyperref}
\usepackage{mathrsfs}
\usepackage{bbm}
\usepackage{amsfonts}
\usepackage{amsmath}
\usepackage{graphicx}
\usepackage{subfig}
\usepackage{feynmp}

\newenvironment{proof}{Proof:}

\def\endofproof {\hfill{$\Box$}\\}

\newcommand{\void}[1]{}

\def\be{\begin{equation}}
\def\ee{\end{equation}}
\def\bea{\begin{eqnarray}}
\def\eea{\end{eqnarray}}
\def\lmd{\lambda}
\def\la{\langle}
\def\ra{\rangle}
\def\bt{\beta}

\newcommand{\td}[1]{\tilde{#1}}

\newcommand{\bary}{\begin{array}}
\newcommand{\eary}{\end{array}}

\def\nb{\nonumber}

\def\bea{\begin{eqnarray}}
\def\eea{\end{eqnarray}}
\def\beq{\begin{equation}}
\def\eeq{\end{equation}}
\def\to{\rightarrow}
\newcommand{\beas}{\begin{eqnarray*}}
\newcommand{\eeas}{\end{eqnarray*}}
\newcommand{\ba}{\begin{array}}
\newcommand{\ea}{\end{array}}

\newcommand{\eref}[1]{{(\ref{#1})}}
\newtheorem{lemma}{Lemma}
\newtheorem{theorem}{Theorem}
\newtheorem{corollary}{Corollary}
\newcommand{\cref}[1]{{{\bf{Corollary~}\ref{#1}}}}
\newcommand{\tref}[1]{{{\bf{Theorem~}\ref{#1}}}}
\newcommand{\sref}[1]{{{\bf{Section~}\ref{#1}}}}
\newcommand{\lref}[1]{{{\bf{Lemma~}\ref{#1}}}}

\begin{document}

   \title{Boundary Behaviors for General Off-shell Amplitudes in Yang-Mills Theory}

\author{Yun Zhang\footnote{zyzhangyun2003@gmail.com} and Gang Chen\footnote{corresponding author: gang.chern@gmail.com}}
{\affiliation{Department of Physics, Nanjing University\\
22 Hankou Road, Nanjing 210093, China
}

\hspace{1cm}
\begin{abstract}
In this article, we analyze the boundary behaviors of pure Yang-Mills amplitudes under adjacent and non adjacent BCFW shifts in Feynman gauge. We introduce reduced vertexes for Yang-Mills fields, prove that these reduced vertexes are equivalent to the original vertexes as for the study of boundary behaviors, which greatly simplifies our analysis of boundary behaviors. Boundary behaviors for adjacent shifts are readily obtained using reduced vertexes. Then we prove a theorem on permutation sum and use it to prove the improved boundary behaviors for non-adjacent shifts. Based on the boundary behaviors, we find that it is possible to generalize BCFW recursion relation to calculate general tree level off shell amplitudes.\end{abstract}


\date{\today}
\maketitle

\section{Introduction}


Recent years, BCFW recursion relation \cite{Britto:2004nj,Britto:2004nc,Britto:2004ap} has been widely used in various quantum field theories. At tree level, the amplitudes in pure Yang-Mills theory are rational functions of external momenta and external polarization vectors in spinor form 
\cite{Parke:1986gb,Xu:1986xb,Berends:1987me,Kosower,Dixon1,Witten1}. According to this, BCFW recursion relation was proposed and developed in \cite{Britto:2004nj,Britto:2004nc,Britto:2004ap}, and then proved in \cite{Britto:2005fq} using the pole structures of the tree level on shell amplitudes.  Besides the progresses on on-shell amplitudes,  off-shell amplitudes are also studied using BCFW or other methods \cite{Feng, Chen1, Chen3, Britto, Chen2, Berends:1987me}.  Although off-shell amplitudes are gauge dependent and usually complicated, they are of great importance in the phenomenological calculations. Moreover, off-shell amplitudes emerge in the construction of on-shell loop level amplitudes.   Hence it is also valuable to get recursion relations for general off-shell amplitudes.

BCFW recursion relation works very well when the amplitudes vanish at large BCFW shift limit. Hence the boundary behaviors of the amplitudes are very important for building up BCFW recursion relation.  Furthermore, improved boundary behaviors also imply new amplitude relations like BCJ relations \cite{Bern:2008qj,Boels,FengJia}.  At tree and loop level Yang-Mills amplitudes, the boundary behaviors were analyzed in \cite{Nima1, Boels} in AHK gauge for both adjacent and non-adjacent BCFW shifts. Hence a natural question is whether it is possible to analyze the boundary behaviors in usual Feynman gauge, and why essentially non-adjacent BCFW shifts have improved boundary behaviors in Feynman gauge comparing with adjacent BCFW shifts. Furthermore, according to the boundary behaviors, can we build up the recursion relation correspondingly for general off-shell amplitudes?

In this article, we first describe the procedure to obtain general off-shell amplitudes recursively in \sref{Sec:Off-shell} using BCFW technique and the technique in \cite{Chen3}. The procedure bases on the boundary behaviors of amplitudes in Feynman gauge, which are proved in the following sections. In \sref{Sec:Reduce} we prove that the boundary behaviors of amplitudes can be analyzed using reduced vertexes, which are defined in the section. Using the conclusion of this section, we directly obtain the boundary behaviors for adjacent shifts. In \sref{Sec:Non-Adj} we analyze the behaviors of the amplitudes for non-adjacent shifts. We find that permutation sum greatly improves the boundary behaviors for non-adjacent shifts compared to adjacent shifts. 


 
\section{Recursion Relation for General Off-shell Amplitudes}\label{Sec:Off-shell}
Throughout this paper, we will use $k_l$ and $k_r$ for the pair of momenta to be shifted, with indices $\mu$ and $\nu$. The momenta shift is
\be
\hat k_l=k_l+z \eta~~~~~~\hat k_r=k_r-z\eta,
\label{momshift}
\ee
with 
\be\label{conditionEta}
\eta^2=k_l\cdot \eta=k_r\cdot \eta=0. 
\ee 
Since we need to shift two off-shell lines for general off-shell amplitudes in Yang-Mills theory, we do not require the momenta of the two shifted lines, ie. $k_l$ and $k_r$, to be on-shell. Other un-shifted lines are also in general off shell. Let two arbitrary vectors $\epsilon_{l\ \mu}$ and $\epsilon_{r\ \nu}$ couple to the two shifted lines, the amplitude is $\mathcal{M}^{\mu\nu} \epsilon_{l\ \mu} \epsilon_{r\ \nu}$. The indices of other external lines are suppressed.

To get all the components of $\mathcal{M}^{\mu\nu}$, we need to know the amplitudes $\mathcal{M}^{\mu\nu} \epsilon_{l\ \mu} \epsilon_{r\ \nu}$ for $4\times 4$ independent pairs of $\epsilon_{l\ \mu}$ and $\epsilon_{r\ \nu}$ in four dimensional field theory.  According to \cite{Chen1,Chen2}, when one of the shifted lines contracts with its momentum, there is a natural recursion relation according to the cancellation details of Ward identity in Feynman gauge. For example with color ordered amplitude $A(k_1,k_2,\cdots,k_{N+1})^{\mu_1\mu_2\cdots\mu_{N+1}}$, we derive:
\begin{small}
\bea
&&k_{\mu_{N+1}} \mathcal{M}(k_1,k_2,\cdots,k_{N+1})^{\mu_1\mu_2\cdots\mu_{N+1}}\label{longitudinalcomp}\\
=&&-\frac{1}{\sqrt{2}}\frac{k_1^2}{(k_1+k_{N+1})^2}\delta^{\mu_1}_\rho \mathcal{M}(k_2,k_3,\cdots,k_{N},k_1+k_{N+1})^{\mu_2\mu_3\cdots\mu_N \rho}\nonumber\\
&&+\frac{1}{\sqrt{2}}\frac{k_N^2}{(k_N+k_{N+1})^2}\delta^{\mu_N}_\rho \mathcal{M}(k_1,k_2,\cdots,k_{N-1},k_1+k_{N})^{\mu_1\mu_2\cdots\mu_{N-1} \rho}\nonumber\\
&&+\sum_{j=1}^{N-1}\frac{i K_{1,j\ \rho}\mathcal{M}(k_1,k_2,\cdots,k_j,K_{1,j})^{\mu_1\mu_2\cdots\mu_j \rho} k_{N+1\ \sigma} \mathcal{M}(k_{j+1},k_{j+2},\cdots,k_N,K_{j+1,N})^{\mu_{j+1}\mu_{j+2}\cdots\mu_N \sigma}}{\sqrt{2}K_{1,j}^2 K_{j+1,N}^2}\nonumber\\
&&-\sum_{j=1}^{N-1}\frac{i k_{N+1\ \rho}\mathcal{M}(k_1,k_2,\cdots,k_j,K_{1,j})^{\mu_1\mu_2\cdots\mu_j \rho} K_{j+1,N\ \sigma} \mathcal{M}(k_{j+1},k_{j+2},\cdots,k_N,K_{j+1,N})^{\mu_{j+1}\mu_{j+2}\cdots\mu_N \sigma}}{\sqrt{2}K_{1,j}^2 K_{j+1,N}^2}.\nonumber
\eea
\end{small}
In the above we have reduced $k_{\mu_{N+1}} \mathcal{M}(k_1,k_2,\cdots,k_{N+1})^{\mu_1\mu_2\cdots\mu_{N+1}}$ to less point amplitudes. $K_{1,j}=k_1+k_2+\cdots+k_j$ and $K_{j+1,N}=k_{j+1}+\cdots+k_N$. The indices for the amplitudes are in the same order as the momenta in the brackets of the amplitudes. In the first two lines, $\delta$ is Kronecker delta. In the last two lines, when $j=1$ or $j=N-1$, we define $\mathcal{M}(k_1,k_1)^{\mu_1 \rho}=i k_1^2 g^{\mu_1\rho}$ and $\mathcal{M}(k_N,k_N)^{\mu_N \sigma}=i k_N^2 g^{\mu_N \sigma}$.

Hence to build up BCFW recursion relation for general off-shell amplitudes, we only need to consider other three components of the external vectors coupling to the shifted lines.

For convenience, the momenta can be written in spinor form:
\be
k=\left\{ \begin{array}{cl}
\lmd\td\lmd+\bt\td\bt & ~~~~~\textrm{if $k$ is time-like}\\
\lmd\td\lmd-\bt\td\bt & ~~~~~\textrm{if $k$ is space-like,}\\
\lmd\td\lmd & ~~~~~\textrm{if $k$ is light-like}
\end{array}\right.
\ee
where the spinors with tilde are the complex conjugates of those without tilde for real momenta. Here we exemplify the cases with time-like or light-like $k_l$ and $k_r$, and the case with either space-like $k_l$ or $k_r$ is similar.

We first consider the case with both $k_l$ and $k_r$ off shell. We write $k_l$ as $k_l= \lmd_l\td\lmd_l+\bt_l\td\bt_l$ \cite{Chalmers}.  As analyzed in \cite{Chen3}, since there is $U(2)$ freedom for choosing the  spinors of $k_l$, we can choose them such that $(\lmd_l\td\bt_l)\cdot k_r=(\bt_l\td\lmd_l)\cdot k_r=0$. At the same time we can set the spinors for $k_r$ to be either $k_r= \lmd_l\td\lmd_r+\bt_r\td\bt_l$ or $k_r= \lmd'_r\td\lmd_l+\bt_l\td\bt'_r$. Hence we have two choices for the shifting momentum $\eta$ as $\eta=\lmd_l\td\bt_l$ or $\eta'=\bt_l\td\lmd_l$, which satisfy the condition \eref{conditionEta}.

First for $\eta=\lmd_l\td\bt_l$, the external vectors are written as 
 \be
\begin{array}{cc}
\epsilon_l\in\left( \begin{array}{cl}
 \epsilon_l^-&= \lmd_l\td\bt_l\\
 \epsilon_l^+&= \bt_l\td\lmd_l\\
  \epsilon_l^{\perp}&= \lmd_l\td\lmd_l-\bt_l\td\bt_l,
\end{array}\right) &~~~~~\epsilon_r\in\left( \begin{array}{cl}
 \epsilon_r^-&= \lmd_l\td\bt_l\\
 \epsilon_r^+&= \bt_r\td\lmd_r\\
  \epsilon_r^{\perp}&= \lmd_l\td\lmd_r-\bt_r\td\bt_l-z\lmd_l\td\bt_l
\end{array}\right).
\ea
\label{epsilonset1}
\ee
Under the momenta shift \eref{momshift}, we have
\bea\label{spinorShift}
\td\lmd_l&\rightarrow& \hat{\td\lmd}_l=\td\lmd_l+z\td\bt_l \nb\\
\bt_r&\rightarrow& \hat\bt_r=\bt_r-z\lmd_l.
\eea
In $\epsilon_r^\perp$, we add the term $-z\lmd_l\td\bt_l$, such that after the momenta shift \eref{spinorShift}, $\hat\epsilon_r^\perp$ is independent of z and still $\hat k_r\cdot\hat\epsilon_r^\perp=0$.



Then for $\eta'=\bt_l\td\lmd_l$, we just replace $\epsilon_r$ with $\epsilon'_r$ which is defined as following:
\be
\epsilon'_r\in\left( \begin{array}{cl}
 \epsilon_r^{'-}&= \lmd'_r\td\bt'_r\\
 \epsilon_r^{'+}&= \bt_l\td\lmd_l\\
  \epsilon_r^{'\perp}&= \lmd'_r\td\lmd_l-\bt_l\td\bt'_r-z\bt_l\td\lmd_l
\end{array}\right).
\label{epsilonset2}
\ee
Under the momenta shift, we have 
\bea\label{spinorShift2}
\lmd_l&\rightarrow& \hat{\lmd}_l=\lmd_l+z\bt_l \nb\\
\td\bt'_r&\rightarrow& \hat{\td\bt}'_r=\td\bt'_r-z\td\lmd_l.
\eea
If one of the lines is on shell and another is off-shell, without loss of generality, we set $l$-line to be on-shell and $r$-line to be off-shell. Writing $k_l$ as $\lmd_l\td\lmd_l$ and using the little group transformation of $k_r$, the momentum of $r$-line can be written as $k_r=\lmd_l\td\lmd_r+\bt_r\td\bt'_r=\lmd'_r\td\lmd_l+\bt{''}_r\td\bt{'''}_r$. Correspondingly, one of the shifting momentum is $\eta=\lmd_l \td\bt'_r$ and the other is $\eta'=\bt^{''}_r\td\lmd_l$. When the shifting momentum is $\eta$, the external vectors  are written as 
 \be
\begin{array}{cc}
\epsilon_l\in\left( \begin{array}{cl}
 \epsilon_l^-&= {\lmd_l\td\bt'_r\over [\td\lmd_l,\td\bt'_r]}\\
 \epsilon_l^+&={\bt_l\td\lmd_l\over \la\bt_l,\lmd_l\ra}\\
\end{array}\right) &~~~~~\epsilon_r\in\left( \begin{array}{cl}
 \epsilon_r^-&= \lmd_l\td\bt'_r\\
 \epsilon_r^+&= \bt_r\td\lmd_r\\
  \epsilon_r^{\perp}&=\lmd_l\td\lmd_r-\bt_r\td\bt'_r-z \lmd_l\td\bt'_r
\end{array}\right) .
\ea
\ee
Under the momenta shift, the spinors transform as 
\bea\label{spinorShift3}
\td\lmd_l&\rightarrow& \hat{\td\lmd}_l=\td\lmd_l+z\td\bt'_r \nb\\
\bt_r&\rightarrow& \hat\bt_r=\bt_r-z\lmd_l.
\eea
When the shifting momentum is $\eta'$, then the external vectors can be written as 
\be
\begin{array}{cc}
\epsilon_l\in\left( \begin{array}{cl}
 \epsilon_l^-&= {\lmd_l\td\bt'_r\over [\td\lmd_l,\td\bt'_r]}\\
 \epsilon_l^+&={\bt''_r\td\lmd_l\over \la\bt''_r,\lmd_l\ra}\\
\end{array}\right) &~~~~~\epsilon'_r\in\left( \begin{array}{cl}
 \epsilon_r^{'-}&= \lmd'_r\td\bt{'''}_r\\
 \epsilon_r^{'+}&= \bt{''}_r\td\lmd_l\\
  \epsilon_r^{'\perp}&=\lmd'_r\td\lmd_l-\bt{''}_r\td\bt{'''}_r-z\bt{''}_r\td\lmd_l
\end{array}\right).
\ea
\ee
Correspondingly, the spinors transform as 
\bea\label{spinorShift4}
\lmd_l&\rightarrow& \hat{\lmd}_l=\lmd_l+z\bt^{''}_r \nb\\
\td\bt^{'''}_r&\rightarrow& \hat{\td\bt}^{'''}_r=\td\bt^{'''}_r-z\td\lmd_l.
\eea

The case with both shifted lines on-shell is discussed in \cite{Boels}.

To use BCFW recursion relation for the full amplitudes, we need to analyze the boundary behaviors for the amplitudes with shifted momenta. We can find for all the cases discussed above, the following conditions hold
\be\label{conditionVector}
\hat k_l\cdot \hat\epsilon_l=\hat k_r \cdot \hat\epsilon_r=0.
\ee

As will be proved in the following sections, under the conditions \eref{conditionEta} and \eref{conditionVector}, we have  
\be\label{behaviorOff}
\mathcal{\hat M}^{\mu\nu}=\left\{ \begin{array}{cl}
z  A_1 g^{\mu\nu}+ A_0 g^{\mu\nu}+ B^{\mu\nu}+\mathcal{O}({1\over z}) &~~~\text{for adjacent shift}\\
 A'_0 g^{\mu\nu}+\mathcal{O}({1\over z}) &~~~\text{for non-adjacent shift}
\end{array}\right. .
\ee
In \eref{behaviorOff}, all the un-shifted and shifted external lines can be off-shell. 


According to \eref{behaviorOff}, we can get the large $z$ scaling behaviors for general off-shell amplitudes $\mathcal{M}^{\mu\nu} \epsilon_{l\ \mu} \epsilon_{r\ \nu}$ for all the BCFW shifts above:
\begin{itemize}
\item Both $k_l$ and $k_r$ off-shell with shifting momentum: $\eta=\lmd_l\td\bt_l$
\bea\label{behaviorTable1}
\left.\begin{array}{c|c|c|c}
 &\epsilon^-_r &\epsilon^+_r&\epsilon_r^{\perp}\\ \hline 
\epsilon^-_l & z^{-1}&z^2&z^0 \\ \hline 
 \epsilon^+_l&z^2 &z^3&z^2 \\ \hline 
\epsilon_l^{\perp}  &z &z^3&z^2 \\ 
\end{array}\right. &~~~~~~\left.\begin{array}{c|c|c|c}
 &\epsilon^-_r &\epsilon^+_r&\epsilon_r^{\perp}\\ \hline 
\epsilon^-_l & z^{-1}&z&z^{-1} \\ \hline 
 \epsilon^+_l&z &z^2&z\\ \hline 
\epsilon_l^{\perp}  &z^{0} &z^2&z \\ 
\end{array}\right. \nb \\
\text{Adjacent}&~~\text{Non-adjacent }
\eea
\item Both $k_l$ and $k_r$ off-shell with shifting momentum $\eta'=\bt_l \td\lmd_l$
\bea\label{behaviorTable2}
\left.\begin{array}{c|c|c|c}
 &\epsilon^-_r &\epsilon^+_r&\epsilon_r^{\perp}\\ \hline 
\epsilon^-_l & z^{3}&z^2&z^2 \\ \hline 
 \epsilon^+_l&z^2 &z^{-1}&z^0 \\ \hline 
\epsilon_l^{\perp}  &z^3 &z&z^2 \\ 
\end{array}\right. &~~~~~~\left.\begin{array}{c|c|c|c}
 &\epsilon^-_r &\epsilon^+_r&\epsilon_r^{\perp}\\ \hline 
\epsilon^-_l & z^{2}&z&z \\ \hline 
 \epsilon^+_l&z &z^{-1}&z^{-1} \\ \hline 
\epsilon_l^{\perp}  &z^2 &z^{0}&z \\ 
\end{array}\right. \nb \\
\text{Adjacent}&~~\text{Non-adjacent }
\eea
\item $k_l$ on-shell and $k_r$ off-shell with shifting momentum $\eta=\lmd_l \td\bt'_r$
\bea\label{behaviorTable3}
\left.\begin{array}{c|c|c|c}
 &\epsilon^-_r &\epsilon^+_r&\epsilon_r^{\perp}\\ \hline 
\epsilon^-_l & z^{-1}&z^2&z^0 \\ \hline 
 \epsilon^+_l&z^2 &z^3&z^2 \\ 
\end{array}\right. &~~~~~~\left.\begin{array}{c|c|c|c}
 &\epsilon^-_r &\epsilon^+_r&\epsilon_r^{\perp}\\ \hline 
\epsilon^-_l & z^{-1}&z&z^{-1} \\ \hline 
 \epsilon^+_l&z &z^2&z \\ 
\end{array}\right. \nb \\
\text{Adjacent}&~~\text{Non-adjacent }
\eea
\item $k_l$ on-shell and $k_r$ off-shell with shifting momentum $\eta'=\bt^{''}_r\td\lmd_l$
\bea\label{behaviorTable4}
\left.\begin{array}{c|c|c|c}
 &\epsilon^-_r &\epsilon^+_r&\epsilon_r^{\perp}\\ \hline 
\epsilon^-_l & z^{3}&z^2&z^2 \\ \hline 
 \epsilon^+_l&z^2 &z^{-1}&z^0 \\ 
\end{array}\right. &~~~~~~\left.\begin{array}{c|c|c|c}
 &\epsilon^-_r &\epsilon^+_r&\epsilon_r^{\perp}\\ \hline 
\epsilon^-_l & z^{2}&z&z \\ \hline 
 \epsilon^+_l&z &z^{-1}&z^{-1} \\ 
\end{array}\right. \nb \\
\text{Adjacent}&~~\text{Non-adjacent }
\eea
\end{itemize}

According to the little group property and the analysis in \cite{Chen3}, and using essentially the same procedures therein, we can construct the BCFW recursion relation for off shell amplitudes. We exemplify the procedure in the case that all external legs are off shell and show how it is reduced to less point amplitudes.

We choose a specific $r$-line, and two non adjacent $l$-lines, ie. $l_1$ and $l_2$. Then we can do two shifts: $l_1$ and $r$ lines, or $l_2$ and $r$ lines. When we shift $l_1$ and $r$ lines, we shift them as in table \ref{behaviorTable1}, and we choose the vectors coupling to $l_1$ as $\epsilon_{l_1}^-=\eta_1=\lambda_{l_1}\tilde \beta_{l_1}$. At the same time we couple to $l_2$ a vector $\epsilon_{l_2}^-=\eta_2=\lambda_{l_2}\tilde \beta_{l_2}$. For choices of $\epsilon_{r(1)}^-$ and $\epsilon_{r(1)}^\perp$ on $r$ line, the two amplitudes:
\beq
\mathcal{M}_{\mu_r\mu_{l_1}\mu_{l_2}} \epsilon_{l_1}^{-\,\mu_{l_1}} \epsilon_{l_2}^{-\,\mu_{l_2}} \epsilon_{r(1)}^{-\,\mu_r} \ \ \mbox{and}\ \ \mathcal{M}_{\mu_r\mu_{l_1}\mu_{l_2}} \epsilon_{l_1}^{-\,\mu_{l_1}} \epsilon_{l_2}^{-\,\mu_{l_2}} \epsilon_{r(1)}^{\perp\,\mu_r}
\label{decoherence1}
\eeq
are of $\mathcal{O}(z^{-1})$, and can be reduced to less point amplitudes using BCFW technique. The subscript $(1)$ in $\epsilon_{r(1)}^-$ or $\epsilon_{r(1)}^\perp$ means that it is for $l_1-r$ shifting. For the same reason when we shift $l_2$ and $r$-lines, we also obtain two amplitudes:
\beq
\mathcal{M}_{\mu_r\mu_{l_1}\mu_{l_2}} \epsilon_{l_1}^{-\,\mu_{l_1}} \epsilon_{l_2}^{-\,\mu_{l_2}} \epsilon_{r(2)}^{-\,\mu_r} \ \ \mbox{and}\ \ \mathcal{M}_{\mu_r\mu_{l_1}\mu_{l_2}} \epsilon_{l_1}^{-\,\mu_{l_1}} \epsilon_{l_2}^{-\,\mu_{l_2}} \epsilon_{r(2)}^{\perp\,\mu_r}
\label{decoherence2}
\eeq
that are of $\mathcal{O}(z^{-1})$, and can be reduced to less point amplitudes using BCFW technique. In the four amplitudes of \eref{decoherence1} and \eref{decoherence2}, the vectors $\epsilon_r$ coupling to $r$-line are correlated with the vectors coupling to $l_1$ or $l_2$, thus we cannot act on $l_1$ or $l_2$ with their little group generators to obtain other components of the amplitudes. However, from the four amplitudes we can solve out $\mathcal{M}_{\mu_r\mu_{l_1}\mu_{l_2}} \epsilon_{l_1}^{-\,\mu_{l_1}} \epsilon_{l_2}^{-\,\mu_{l_2}}$, such that we can couple $\epsilon_r$ to $r$-line independent of the vectors $\epsilon_{l_1}$ and $\epsilon_{l_2}$ in four dimensional spacetime. Then we can act on $\mathcal{M}_{\mu_r\mu_{l_1}\mu_{l_2}} \epsilon_{l_1}^{-\,\mu_{l_1}} \epsilon_{l_2}^{-\,\mu_{l_2}}$ with the little group generators for $l_1$ and $l_2$ lines, and get all of $\mathcal{M}_{\mu_r\mu_{l_1}\mu_{l_2}} \epsilon_{l_1}^{i\,\mu_{l_1}} \epsilon_{l_2}^{j\,\mu_{l_2}}$ with $i,j\in\{-,\perp,+\}$. Together with the longitudinal components which have been reduced to less point amplitudes in \eref{longitudinalcomp}, we have set up a BCFW recursion relation for general off shell amplitudes.

Several supplements for the above procedure. First, if for some special cases, \eref{decoherence1} and \eref{decoherence2} cannot determine $\mathcal{M}_{\mu_r\mu_{l_1}\mu_{l_2}} \epsilon_{l_1}^{-\,\mu_{l_1}} \epsilon_{l_2}^{-\,\mu_{l_2}}$, we can replace either $l_1-r$ shift or $l_2-r$ shift as in Table \ref{behaviorTable2}. Second, when one of the shifted lines is on shell, we can get the $\epsilon^-$ and $\epsilon^+$ components on this on shell line using the above procedure, and the momentum component from \eref{longitudinalcomp}. These components are sufficient for an on shell line. Third, in the above procedure, we required $l_1$ and $l_2$ both non-adjacent to $r$ line. Actually for the procedure to work, we only need three amplitudes which can be reduced by BCFW technique, with the fourth amplitude from \eref{longitudinalcomp}. From Table \ref{behaviorTable1} or \ref{behaviorTable2}, we can see that a non-adjacent shift plus an adjacent shift is already enough for the procedure to work, which means that our procedure works from 4 point level.

In conclusion, with the proper boundary behaviors to be discussed in the following sections, and using the little group techniques in \cite{Chen3}, BCFW recursion relation can be generalized to calculate general tree level amplitudes with any number of off shell lines.

\section{Amplitudes with Reduced Vertexes}\label{Sec:Reduce}

In this section we are going to introduce some reduced vertexes for the ordinary color ordered Yang-Mills vertexes, and prove that amplitudes constructed from the reduced vertexes have the same boundary behaviors as those constructed from ordinary vertexes.


We first clarify some conventions for the rest of this article. If we draw the complex momentum line from left to right, other external legs besides the shifted pair would be either above or below this complex line.  For a given shift, the set of external legs above(or below) the complex line is fixed together with their order, however the legs above the complex line and those below it can have all possible relative positions. To further specify the vertexes, we sort the vertexes as in Figure \ref{vertexclass}.
\begin{figure}[htb]
\centering
\includegraphics{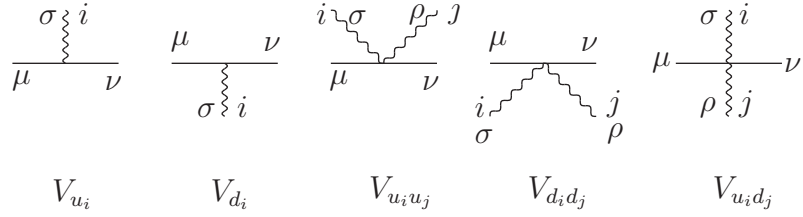}
\caption{A classification of the vertexes. The horizontal line is the complex line and the photon lines are external legs besides the shifted pair. i and j are the numbering of the external states, and $\mu$, $\nu$, $\sigma$, $\rho$ are the indices.}
\label{vertexclass}
\end{figure}

For a three-point vertex with line 1, 2 and 3 in anti-clockwise order, we write it in the following form:
\bea\label{newV3}
V_{\mu_1\mu_2\mu_3}&\equiv&S_{\mu_1\mu_2\mu_3}+ R_{\mu_1\mu_2\mu_3}+M_{\mu_1\mu_2\mu_3},
\eea
where  
\bea\label{newV3s}
S_{\mu_1\mu_2\mu_3}&=&\frac{i}{\sqrt 2}\left(g_{\mu_1\mu_2}(k_1-k_2)_{\mu_3}\right) \nb\\
R_{\mu_1\mu_2\mu_3}&=&\frac{i}{\sqrt 2}\left(-2g_{\mu_2\mu_3}(k_3)_{\mu_1}+2g_{\mu_3\mu_1}(k_3)_{\mu_2}\right) \nb\\
M_{\mu_1\mu_2\mu_3}&=&\frac{i}{\sqrt 2}\left(-g_{\mu_2\mu_3}(k_1)_{\mu_1}+g_{\mu_3\mu_1}(k_2)_{\mu_2}\right). 
\eea
In this manner, $k_3$ is in a special role and we will choose the appropriate one as $k_3$ in specific situations. When the lines 1 and 2 are on the complex line and 3 is an external leg, we further divide the M term into $M^L$ and $M^R$ as represented in Figure \ref{Msymbol}.
\begin{figure}[]
\centering
\includegraphics[height=4cm,width=12cm]{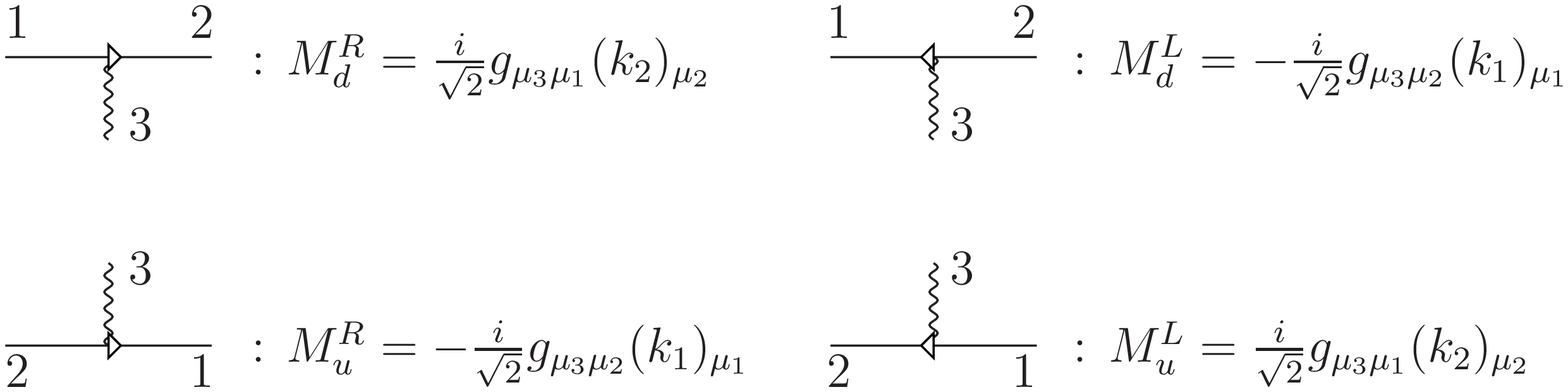}
\caption{The symbols and meanings of $M_{u/d}^{L/R}$.}
\label{Msymbol}
\end{figure}

Contracting a three point vertex $V_{\mu_1\mu_2\mu_3}$ with $k_3^{\mu_3}$, we get:
\beq
k_3^{\mu_3} \cdot V_{\mu_1\mu_2\mu_3}=\frac{i}{\sqrt{2}}g_{\mu_1\mu_2} k_2^2-\frac{i}{\sqrt{2}}g_{\mu_1\mu_2} k_1^2-\frac{i}{\sqrt{2}}k_{2\ \mu_2}k_2{}_{\ \mu_1}+\frac{i}{\sqrt{2}}k_{1\ \mu_1}k_1{}_{\ \mu_2},
\label{kdotV}
\eeq
and we represent these terms by the symbols in Figure \ref{vertexnotation}.
\begin{figure}[htb]
\centering
\includegraphics[height=4cm,width=10cm]{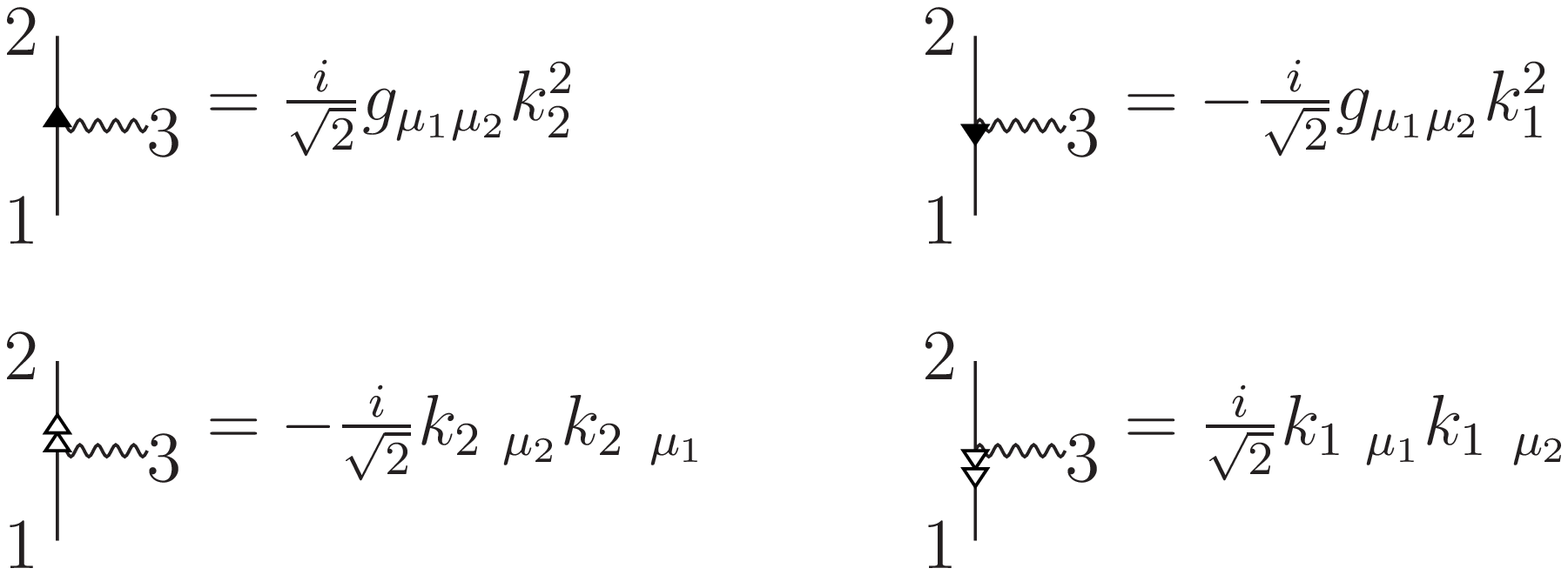}
\caption{Notations for \eref{kdotV}.}
\label{vertexnotation}
\end{figure}

In the following of this paper, the method of induction is assumed. For example, when we discuss the $\mathcal{O}(z^1)$, $\mathcal{O}(z^0)$ and $\mathcal{O}(z^{-1})$ behaviors of N point amplitudes, we only need to consider the diagrams with all the external legs attaching the complex line. When some of these external legs form vertexes outside the complex line, it is not changed whether the shift is adjacent or non-adjacent, and the conclusions for less external leg amplitudes apply to these diagrams when we do not require the external legs to be on shell.


\subsection{Reduced Vertexes}
The central conclusion of this subsection is that the boundary behaviors for BCFW momenta shift \eref{momshift} under the conditions in \eref{conditionEta} and \eref{conditionVector} can be obtained by using the reduced vertexes  as following:
\bea
{\bar V}_{u/d}&=&S_{u/d}+R_{u/d},\nonumber\\
{\bar V}_{u_i u_j}&=&\frac{i}{2}(2 g^{\nu\sigma}g^{\mu\rho}-2g^{\mu\sigma}g^{\nu\rho}-g^{\sigma\rho}g^{\mu\nu}),\nonumber\\
{\bar V}_{d_i d_j}&=&\frac{i}{2}(2g^{\nu\sigma}g^{\mu\rho}-2g^{\mu\sigma}g^{\nu\rho}-g^{\sigma\rho}g^{\mu\nu}),\nonumber\\
{\bar V}_{u_i d_j}&=&ig^{\sigma\rho}g^{\mu\nu}.
\label{ReduV}
\eea
The meanings of the vertex names, the external legs and their indices refer to Figure \ref{vertexclass}, and the meanings of S term and R term in the first line refer to \eref{newV3s} with the external leg playing the role of Line 3. 

We first prove some useful lemmas. First, for a tree level tensor current $\mathcal{M}_{12\cdots N}^{\mu_1\mu_2\cdots \mu_N}$, we shift $k_i$ and $k_j$: $\hat k_i\to k_i+z\eta$ and $\hat k_j\to k_j-z\eta$ with $\eta^2=0$ and $k_i\cdot \eta=0$. We couple $\hat \epsilon_i$ to the $\hat k_i$ line with $\hat k_i \cdot \hat \epsilon_i=0$. If $\hat \epsilon_i\sim \mathcal{O}(z^{n_i})$, naive power counting gives ${\hat k}_{j\ \mu_j} \hat{\mathcal{M}}_{12\cdots N}^{\mu_1\mu_2\cdots \mu_N} \hat \epsilon_{i\ \mu_i}\sim \mathcal{O}(z^{2+n_i})$. However, we have:

\begin{lemma}{Generalized Ward Identity 1}
\\${\hat k}_{j\ \mu_j} \hat{\mathcal{M}}_{12\cdots N}^{\mu_1\mu_2\cdots \mu_N} \hat \epsilon_{i\ \mu_i}\sim \mathcal{O}(z^{n_i})$, for $\hat k_i\to k_i+z\eta$ and $\hat k_j\to k_j-z\eta$ with $\eta^2=0$, $k_i\cdot \eta=0$ and $\hat k_i \cdot \hat \epsilon_i=0$.
\label{GWI1}
\end{lemma}
\begin{proof}
The proof can be done by induction, similar to the proof of actual tree-level Ward identity in our other papers \cite{Chen1,Chen2}. For three point tensor currents this Lemma can be verified directly. Assume it holds for no more than N point tensor currents. We construct an (N+1) point tensor current by inserting $k_j$ into an N-point one. Those diagrams with some external legs not attaching the complex line directly need not be considered since they apply the results for no more than N point tensor currents.

When $k_j$ is inserted on a propagator or external leg to form a three vertex $V_j$, we use the notations in Figure \ref{vertexnotation} to decompose $k_j\cdot V_j$. Among the four terms, the first line two terms, ie. solid triangle terms, plus the terms from $k_j$ inserted to a three point vertex in the N point diagram cancel as in Figure \ref{treecancel2}.
\begin{figure}[]
\centering
\includegraphics{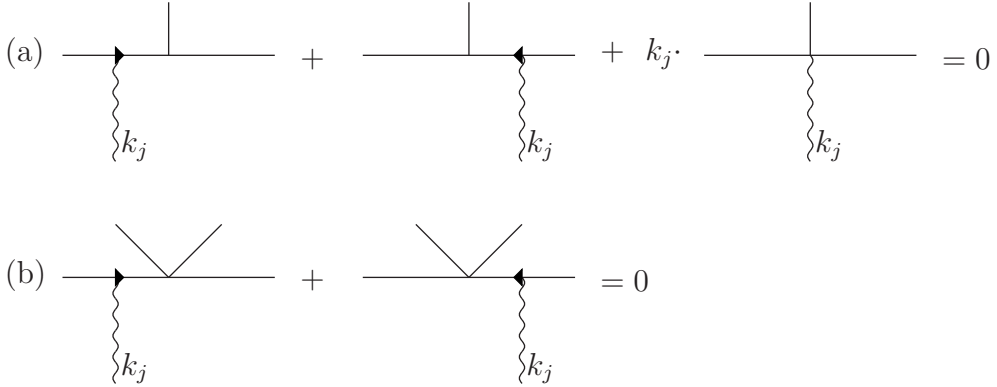}
\caption{Groups of terms that cancel. These terms cancel solely due to the vertex, without any on shell conditions on the legs. Only the $k_j$ line as in \lref{GWI1} are specially represented by photon line. This Figure is from \cite{Chen2}.} 
\label{treecancel2}
\end{figure}

Then the remaining terms are the second line double hollow triangle terms in Figure \ref{vertexnotation} when $k_j$ is inserted to a propagator or external leg in the original N point diagram. Then by direct power counting or the use of the induction assumption, it is seen that the order of z are decreased by at least 2. Thus, we have proven that for N+1 point amplitude, the order of z for ${\hat k}_{j\ \mu_j} \hat{\mathbf{M}}_{12\cdots N}^{\mu_1\mu_2\cdots \mu_N} \hat \epsilon_{i\ \mu_i}$ are decreased by at least 2 from naive power counting, finishing the proof for \lref{GWI1}. \endofproof
\end{proof}

\begin{lemma}{Generalized Ward Identity 2}
\\${\hat k}_{j\ \mu_j} {\hat k}_{i\ \mu_i} \mathcal{M}_{12\cdots N}^{\mu_1\mu_2\cdots \mu_N}\sim \mathcal{O}(z^1)$, for a shift: $\hat k_i\to k_i+z\eta$ and $\hat k_j\to k_j-z\eta$ with $\eta^2=0$.
\label{GWI2}
\end{lemma}
In this Lemma, no on shell condition is placed on leg i or j. ${\hat k}_{j\ \mu_j} {\hat k}_{i\ \mu_i} \mathcal{M}_{12\cdots N}^{\mu_1\mu_2\cdots \mu_N}\sim \mathcal{O}(z^3)$ by naive power counting, yet decreased by 2 orders of z. This Lemma can also be proved by induction with the same procedure as the proof for the above Lemma.

With the above two Lemmas, we are ready to prove our central conclusion \tref{redamplitude} of this subsection.

For each diagram the vertexes in it are \{$V_{u_i}$, $V_{d_j}$, $V_{u_i u_{i+1}}$, $V_{d_j d_{j+1}}$, $V_{u_i d_j}$\}, determined by the different orderings of the external legs. We denote this diagram as $\mathcal{M}^{\mu\nu}(\{u_i,d_j, u_i u_{i+1}, d_j d_{j+1}, u_i d_j\})$. In the rest of the article, and also for \eref{behaviorOff}, when we talk about $\hat{\mathcal{M}}^{\mu\nu}(\{u_i,d_j, u_i u_{i+1}, d_j d_{j+1}, u_i d_j\})$ with $\mu$ and $\nu$ indices not contracted with other tensors, we will always assume it contracted with $\hat \epsilon^\mu_l\sim \mathcal{O}(z^{n_l})$ and $\hat \epsilon^\nu_r\sim \mathcal{O}(z^{n_r})$, which satisfy $\hat k_l \cdot \hat \epsilon_l=0$ and $\hat k_r \cdot \hat \epsilon_r=0$, and we will not write $\hat \epsilon^\mu_l$ and $\hat \epsilon^\nu_r$, and suppress $n_l+n_r$ in the order z analysis of the amplitudes.

\begin{theorem} 
For the shift of a pair of momenta $\hat k_l^\mu=k_l^\mu+z \eta^\mu$ and $\hat k_r^\nu=k_r^\nu-z\eta^\nu$, the amplitude at large z has the property:
\beq
\hat{\mathcal{M}}^{\mu\nu}(\{u_i,d_j, u_i u_{i+1}, d_j d_{j+1}, u_i d_j\})=\hat{\mathcal{M}}^{\mu\nu}(\{\overline{u_i},\overline{d_j}, \overline{u_i u_{i+1}}, \overline{d_j d_{j+1}}, \overline{u_i d_j}\})+\mathcal{O}(z^{-1}).
\eeq
\label{redamplitude}
\end{theorem}

$(\{\overline{u_i},\overline{d_j}, \overline{u_i u_{i+1}}, \overline{d_j d_{j+1}}, \overline{u_i d_j}\})$ means that the vertexes are the reduced vertexes $(\{\bar V_{u_i},\bar V_{d_j}, \bar V_{u_i u_{i+1}}, \bar V_{d_j d_{j+1}}, \bar V_{u_i d_j}\})$ respectively. The highest possible scaling behavior for $\hat{\mathcal{M}}^{\mu\nu}(\{u_i,d_j, u_i u_{i+1}, d_j d_{j+1}, u_i d_j\})$ is $\mathcal{O}(z^1)$, and this theorem says that the first two orders are determined by the reduced vertexes. The reduced vertexes refer to \eref{ReduV}.

\begin{proof}
{\bf Step 1.} We notate a diagram by the positions of the vertexes from left to right on the complex line. Using $V_{u/d}=\bar V_{u/d}+M_{u/d}^L+M_{u/d}^R$, we have:
\bea
&&{\bar V}_{(u/d)_1} {\bar V}_{(u/d)_2}\cdots{\bar V}_{(u/d)_n}\\
=&&(V_{(u/d)_1}-M_{(u/d)_1}^L-M_{(u/d)_1}^R)(V_{(u/d)_2}-M_{(u/d)_2}^L-M_{(u/d)_2}^R)\cdots(V_{(u/d)_n}-M_{(u/d)_n}^L-M_{(u/d)_n}^R),\nonumber
\eea
and by expanding it we get:
\bea
&&V_{(u/d)_1} V_{(u/d)_2}\cdots V_{(u/d)_n}\nonumber\\
=&&{\bar V}_{(u/d)_1} {\bar V}_{(u/d)_2}\cdots{\bar V}_{(u/d)_n}\nonumber\\
+&&\sum_i (-1)^{i-1} \mbox{(i vertexes are replaced with their M term components}).
\label{expansion}
\eea
For diagrams containing four point vertexes, we only re-express the three point vertexes therein without any change to four point vertexes at this step, and then do the similar expansion as in \eref{expansion}.

{\bf Step 2.} In this step, we prove that for each term in \eref{expansion}, in order to contribute at $\mathcal{O}(z^1)$ and $\mathcal{O}(z^0)$, the last M factor in the term should be $M_{u/d}^L$, and for the same reason the first M factor should be $M_{u/d}^R$. This is clearly shown in (a) in Figure \ref{Mtermcancel}. 

\begin{figure}[]
\includegraphics[height=12cm,width=12cm]{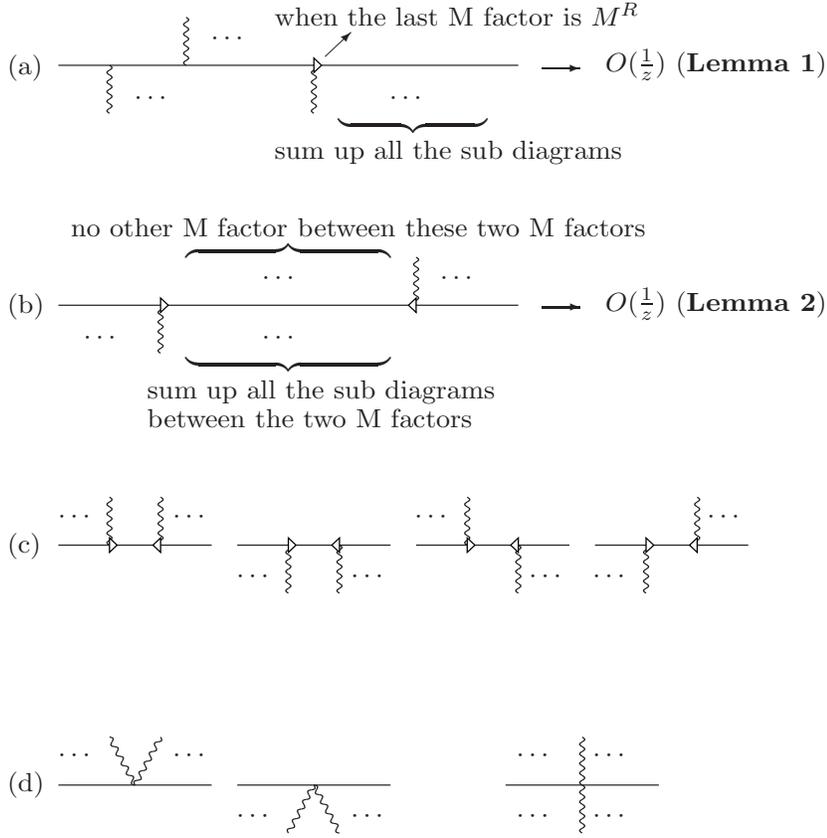}~~
\caption{In (a), naive power counting tells us that any diagram is $\mathcal{O}(z^{1})$, yet the sum turns out to be $\mathcal{O}(z^{-1})$ since \lref{GWI1} tells us that for this sum the actual z dependence is at least lowered by 2 orders compared to naive power counting. The same manner works for the case when the last M factor is $M_u^R$ instead of $M_d^R$, and also works for the analysis of the first M factor in the terms of \eref{expansion}. In (b) \lref{GWI2} works when and only when there are some vertexes between the two M factors. When the two M factors are next to each other, the contributions are shown in (c), which escape \lref{GWI2} and may contribute to $\mathcal{O}(z^0)$. The vertexes besides these two M terms are summed up to be reduced three point vertexes as explained in {\bf Step 3}. (d) gives the corresponding terms which add up with (c) to replace the four point vertexes with the reduced ones.}
\label{Mtermcancel}
\end{figure}

{\bf Step 3.} From step 2 we know that for contributions at $\mathcal{O}(z^{1})$ and $\mathcal{O}(z^0)$, for all the terms containing at least one M factor in \eref{expansion}, we only need to consider the terms where the last M factor is $M_{u/d}^L$ and the first M factor is $M_{u/d}^R$. For such terms, there clearly exists a pair of M factors $M_{u/d}^R$ and $M_{u/d}^L$, where $M_{u/d}^R$ is on the left of $M_{u/d}^L$ and there are no other M factors between them. This is represented in (b) of Figure \ref{Mtermcancel}. Due to \lref{GWI2} these terms do not contribute to $\mathcal{O}(z^{1})$ and $\mathcal{O}(z^0)$, except the special terms represented in (c) of Figure \ref{Mtermcancel}, where the $M_{u/d}^R$ and $M_{u/d}^L$ are next to each other. For the terms in (c), since the product of the two M terms decrease the order of z by 1, there can be no other four point vertexes at the two sides of the two M terms, in order to contribute to $\mathcal{O}(z^{1})$ and $\mathcal{O}(z^0)$. The terms in (c) add up to be:
\beq
{\bar V}_{(u/d)_1} {\bar V}_{(u/d)_2}\cdots{\bar V}_{(u/d)_i} M_{(u/d)_{i+1}}^R M_{(u/d)_{i+2}}^L{\bar V}_{(u/d)_{i+3}}\cdots{\bar V}_{(u/d)_n},
\eeq
which means that on the two sides of the two M terms all the vertexes are the reduced three point vertexes \eref{ReduV}.

{\bf Step 4.} In the first 3 steps, we have analyzed the terms in \eref{expansion} with at least one M factor, which are reduced to the terms in (c) of Figure \ref{Mtermcancel}. The other terms in \eref{expansion} are either all comprised of reduced three point vertexes, or of reduced three point vertexes plus one and only one four point vertex. The latter case is given in (d) of Figure \ref{Mtermcancel}. (c) and (d) sum up to replace the four point vertex with the reduced one. Thus, we have shown that at $\mathcal{O}(z^{1})$ and $\mathcal{O}(z^0)$, all the terms in \eref{expansion} are reduced either to a product of reduced three point vertexes, or a product of reduced three point vertexes and one reduced four point vertex.\endofproof
\end{proof}

\subsection{Application}\label{theorem1Application}
As a simple application of \tref{redamplitude}, we can directly obtain the large-z scaling behaviors for amplitudes with adjacent BCFW shifts. 

For $\mathcal{M}^{\mu\nu}(\{u_i,d_j, u_i u_{i+1}, d_j d_{j+1}, u_i d_j\})$, we denote the product of all the vertexes in it as $\mathcal{N}^{\mu\nu}(\{u_i,d_j, u_i u_{i+1}, d_j d_{j+1}, u_i d_j\})$, and the product of all the propagators in the complex line in it as $\mathcal{C}(\{u_i,d_j, u_i u_{i+1}, d_j d_{j+1}, u_i d_j\})$.  Here and following, we usually suppress $(\{u_i,d_j, u_i u_{i+1}, d_j d_{j+1}, u_i d_j\})$ for convenience. 
Then  the amplitude is written as 
\bea\label{infyTerm}
\mathcal{M}^{\mu\nu}=\sum_{\mathcal{D}}{\mathcal{N}^{\mu\nu}\over \mathcal{C}} , 
\eea
where the sum is over all the Feynman diagrams.
 
The amplitude can be expanded as $\mathcal{M}^{\mu\nu}=\mathcal{M}_1^{\mu\nu} z + \mathcal{M}_0^{\mu\nu}+\mathcal{M}^{\mu\nu}_{-1} {1\over z}+ \mathcal{O}({1\over z^2})$ in the large $z$ limit. We need to discuss the large-z scaling behaviors for some types of Feynman diagrams. For convenience we denote the types of Feynman diagrams as following: $\mathcal{D}_I$ denotes the diagrams where all vertexes in the complex line are reduced three point vertexes. $\mathcal{D}_{II}$ denotes the diagrams where the complex line contains only one reduced four point vertex which is not $\bar V_{u_id_j}$ and other vertexes are reduced three point vertexes, while in $\mathcal{D'}_{II}$ the four point vertex is $\bar V_{u_id_j}$. In $\mathcal{D}_{III}$, there are two reduced four point vertexes in the complex line neither of which is $\bar V_{u_id_j}$ and other vertexes are reduced three point vertexes, while in $\mathcal{D'}_{III}$ at least one of the four point vertexes is $\bar V_{u_id_j}$. For $\mathcal{M}_1^{\mu\nu}$ and $\mathcal{M}_0^{\mu\nu}$, we only need to take $\mathcal{D}_I$, $\mathcal{D}_{II}$ and $\mathcal{D'}_{II}$ into consideration.

The contribution to the amplitudes from each kind of Feynman diagrams can be expanded respectively as:
\bea
{\mathcal{N}_{h_I}^{\mu\nu} z^{h_I}+ \mathcal{N}_{{h_I}-1}^{\mu\nu} z^{{h_I}-1}+\mathcal{N}_{{h_I}-2}^{\mu\nu}z^{{h_I}-2}\cdots\over \mathcal{C}_{{h_I}-1} z^{{h_I}-1} +\mathcal{C}_{{h_I}-2} z^{{h_I}-2}+\mathcal{C}_{{h_I}-3} z^{{h_I}-3}\cdots} ~&&\text{\ \ for} ~\mathcal{D}_I \\
{\mathcal{N}_{h_{II}}^{\mu\nu} z^{h_{II}}+ \mathcal{N}_{{h_{II}}-1}^{\mu\nu} z^{{h_{II}}-1}+\mathcal{N}_{{h_{II}}-2}^{\mu\nu}z^{{h_{II}}-2}\cdots\over \mathcal{C}_{{h_{II}}} z^{h_{II}} +\mathcal{C}_{h_{II}-1} z^{h_{II}-1}+\mathcal{C}_{h_{II}-2} z^{h_{II}-2}\cdots} ~&&\text{\ \ for}~ \mathcal{D}_{II} \\
{\mathcal{N}_{h_{II'}}^{\mu\nu} z^{h_{II'}}+ \mathcal{N}_{{h_{II'}}-1}^{\mu\nu} z^{{h_{II'}}-1}+\mathcal{N}_{{h_{II'}}-2}^{\mu\nu}z^{{h_{II'}}-2}\cdots\over \mathcal{C}_{{h_{II'}}} z^{h_{II'}} +\mathcal{C}_{h_{II'}-1} z^{h_{II'}-1}+\mathcal{C}_{h_{II'}-2} z^{h_{II'}-2}\cdots} ~&&\text{\ \ for}~ \mathcal{D}_{II'} \\
{\mathcal{N}_{h_{III}}^{\mu\nu} z^{h_{III}}+ \mathcal{N}_{h_{III}-1}^{\mu\nu} z^{h_{III}-1}+\mathcal{N}_{h_{III}-2}^{\mu\nu}z^{h_{III}-2}\cdots\over \mathcal{C}_{h_{III}+1} z^{h_{III}+1} +\mathcal{C}_{h_{III}} z^{h_{III}}+\mathcal{C}_{{h_{III}}-1} z^{{h_{III}}-1}\cdots} ~&&\text{\ \ for}~ \mathcal{D}_{III}\\
{\mathcal{N}_{h_{III'}}^{\mu\nu} z^{h_{III'}}+ \mathcal{N}_{h_{III'}-1}^{\mu\nu} z^{h_{III'}-1}+\mathcal{N}_{h_{III'}-2}^{\mu\nu}z^{h_{III'}-2}\cdots\over \mathcal{C}_{h_{III'}+1} z^{h_{III'}+1} +\mathcal{C}_{h_{III'}} z^{h_{III'}}+\mathcal{C}_{{h_{III'}}-1} z^{{h_{III'}}-1}\cdots} ~&&\text{\ \ for}~ \mathcal{D}_{III'}
\eea 
where we use $\mathcal{N}_{h_I}^{\mu\nu}$ et al. to denote the highest $z$-order term of $\mathcal{N}^{\mu\nu}$ for each Feynman diagram.

Then we can write 
\bea\label{LargeExap}
\mathcal{M}_1^{\mu\nu}&=&\mathcal{\bar M}_1^{\mu\nu}\nb\\
\mathcal{M}_0^{\mu\nu}&=&\mathcal{\bar M}_0^{\mu\nu}+\sum_{\mathcal{D}_{II'}} {\mathcal{N}_{h_{II'}}^{\mu\nu}\over \mathcal{C}_{h_{II'}}}-\sum_{\mathcal{D}_I} {\mathcal{C}_{h_I-2} \mathcal{N}_{h_I}^{\mu\nu}\over \mathcal{C}_{h_I-1}^2 }\nb\\
\mathcal{M}_{-1}^{\mu\nu}&=&\mathcal{\bar M}_{-1}^{\mu\nu}-\sum_{\mathcal{D}_1}{\mathcal{C}_{h_I-2} \mathcal{N}_{h_I-1}^{\mu\nu}\over \mathcal{C}_{h_I-1}^2}+\sum_{\mathcal{D}_I} {(\mathcal{C}_{h_I-2}^2-\mathcal{C}_{h_I-1}\mathcal{C}_{h_I-3})\mathcal{N}_{h_I}^{\mu\nu}\over \mathcal{C}_{h_I-1}^3}-\sum_{\mathcal{D}_{II}}{\mathcal{C}_{h_{II}-1} \mathcal{N}_{h_{II}}^{\mu\nu}\over \mathcal{C}_{h_{II}}^2}\nb\\
&&-\sum_{\mathcal{D}_{II'}}{\mathcal{C}_{h_{II'}-1} \mathcal{N}_{h_{II'}}^{\mu\nu}\over \mathcal{C}_{h_{II'}}^2}+\sum_{\mathcal{D}_{II'}} {\mathcal{N}_{h_{II'}-1}^{\mu\nu}\over \mathcal{C}_{h_{II'}}}+\sum_{\mathcal{D}_{III'}}
{\mathcal{N}_{h_{III'}}^{\mu\nu}\over \mathcal{C}_{h_{III'}+1}}+\mathcal{M}_{-1(M)}^{\mu\nu},
\eea
with
\bea\label{barTerm}
\mathcal{\bar M}_1^{\mu\nu}&=&\sum_{\mathcal{D}_I} {\mathcal{N}_{h_I}^{\mu\nu}\over \mathcal{C}_{h_I-1}}\nb\\
\mathcal{\bar M}_0^{\mu\nu}&=&\sum_{\mathcal{D}_1} {\mathcal{N}_{h_I-1}^{\mu\nu}\over \mathcal{C}_{h_I-1}}+\sum_{\mathcal{D}_{II}} {\mathcal{N}_{h_{II}}^{\mu\nu}\over \mathcal{C}_{h_{II}}}\nb\\
\mathcal{\bar M}_{-1}^{\mu\nu}&=&\sum_{\mathcal{D}_I} {\mathcal{N}_{h_I-2}^{\mu\nu}\over \mathcal{C}_{h_I-1}}+\sum_{\mathcal{D}_{II}} {\mathcal{N}_{h_{II}-1}^{\mu\nu}\over \mathcal{C}_{h_{II}}}+\sum_{\mathcal{D}_{III}} {\mathcal{N}_{h_{III}}^{\mu\nu}\over \mathcal{C}_{h_{III}+1}}.
\eea
In \eref{LargeExap} the last term for $\mathcal{M}_{-1}^{\mu\nu}$, ie $\mathcal{M}_{-1(M)}^{\mu\nu}$, is the contribution from M terms of the three point vertexes,  which is represented by the diagrams (a) and (b) in Figure \ref{Mtermcancel}. This term will be discussed in \sref{order1overz}. In \eref{LargeExap} and \eref{barTerm}, the summations are over ordered product $OP\{\alpha_{u_N} \bigcup \alpha_{d_M}\}$ \cite{Boels}, where $\alpha_{u_N}$ is the ordered subsets of $N$ up-legs $\{u_1,u_2, \cdots, u_N\}$ and $\alpha_{d_M}$ is the ordered subsets of $M$ down-legs $\{d_1,d_2, \cdots, d_M\}$. The ordered product  is the set of all permutations which leave the order of $\alpha_{u_N}$ and $\alpha_{d_M}$ invariant. For example, we have
\be
\sum_{\mathcal{D}_I} {\mathcal{N}_{h_I}^{\mu\nu}\over \mathcal{C}_{h_I-1}}\equiv\sum_{OP\{\alpha_{u_N} \bigcup \alpha_{d_M}\}} {\mathcal{N}_{h_I}^{\mu\nu}\over \mathcal{C}_{h_I-1}}.
\ee

Using \tref{redamplitude}, we can classify the terms that contribute to $\mathcal{M}_1^{\mu\nu}$ and $\mathcal{M}_0^{\mu\nu}$ into the following groups:
\begin{enumerate}
\item $\mathcal{D}_I$ with all the reduced three point vertexes  taking their S term components.
\item $\mathcal{D}_I$ with only one of the reduced three point vertex taking its R term part.
\item $\mathcal{D}_{II}$ with all the reduced three point vertexes  taking their S term components.
\item $\mathcal{D}_{II'}$ with all the reduced three point vertexes  taking their S term components.
\end{enumerate}
For the meaning of R and S terms in the reduced three point vertexes, refer to \eref{newV3s} and \eref{ReduV}, with the external legs playing the role of Line 3 therein.

Case 1 is manifestly proportional to $g^{\mu\nu}$ and contributes to $\mathcal{N}_{h_I}^{\mu\nu}$ and $\mathcal{N}_{h_I-1}^{\mu\nu}$ in \eref{LargeExap} and \eref{barTerm}; Case 2 and Case 3 contribute to $\mathcal{N}_{h_I-1}^{\mu\nu}$ and $\mathcal{N}_{h_{II}}^{\mu\nu}$ respectively, and are manifestly antisymmetric in $\mu$ ad $\nu$; Case 4, which contributes to $\mathcal{N}_{h_{II'}}^{\mu\nu}$, is manifestly proportional to $g^{\mu\nu}$. Thus according to \eref{LargeExap}, an immediate conclusion is made that, for adjacent or non-adjacent BCFW shifts, $\mathcal{M}_1^{\mu\nu}$ is proportional to $g^{\mu\nu}$, and $\mathcal{M}_0^{\mu\nu}$  is in the form of $A g^{\mu\nu}+B^{\mu\nu}$  with $B^{\mu\nu}$ antisymmetric in $\mu$ and $\nu$. In the next section, we will see how non-adjacent shifts imply improved boundary behaviors compared with adjacent shifts.

\section{Amplitudes for Non-adjacent BCFW Shifts}\label{Sec:Non-Adj}
We first show a property which is special for non-adjacent BCFW shifts. Such property is very useful in analyzing each summation in the right hand side of \eref{barTerm}. Furthermore, it is this property that results in better boundary behaviors for amplitudes under non adjacent shifts.   



\subsection{Permutation Sums}
In this subsection, we discuss $\sum_{\mathcal{D}_I} {\mathcal{N}_{h_I}^{\mu\nu}\over \mathcal{C}_{h_I-1}}$ in detail. The conclusions also hold for other summations in \eref{barTerm}. We use $k_{l,u_i}$ to denote for $k_l+k_{u_1}+k_{u_{u_1}}+\cdots+k_{u_i}$ and $k_{d_j,u_i}$ for $k_{d_j}+k_{d_{j-1}}+\cdots +k_{d_1}+k_l+k_{u_1}+k_{u_2}+\cdots+k_{u_j}$.  As a warm-up exercise, we investigate an example with N legs above and 1 leg below the complex line, see Figure \ref{example1}.
\begin{figure}[]
\centering
\includegraphics{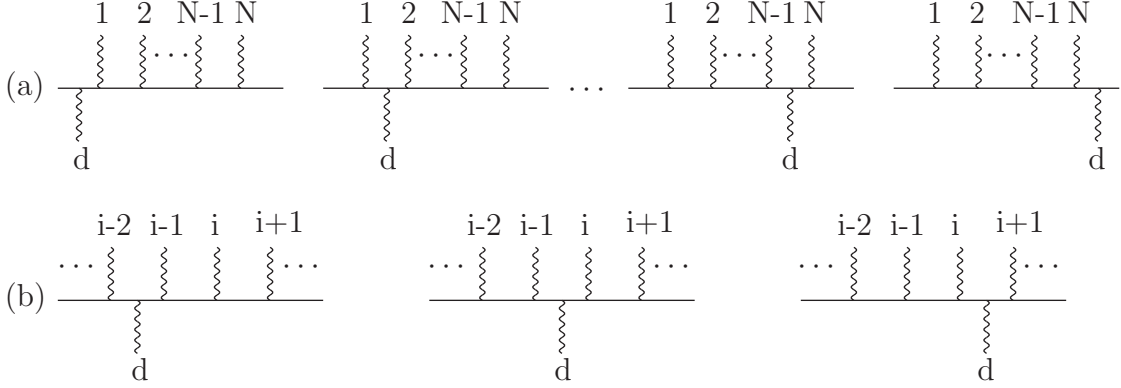}
\caption{When there are N legs above and 1 leg below the complex line, (a) shows all the diagrams with all external legs attaching the complex line and contributing to $\mathcal{O}(z^1)$. (b) contains three diagrams out of (a) for analysis.}
\label{example1}
\end{figure}
We first investigate the highest z order terms of the products of the propagators for the three diagrams as in (b) of Figure \ref{example1}. For convenience, we will omit the $-i$ factors in the propagators in the following. Since there is only one leg "d" below the complex line, this "d" can be viewed as "$d_1$". For the three diagrams of (b) in Figure \ref{example1}, $\mathcal{C}_{h_I-1} (\{u_i,d_j, u_i u_{i+1}, d_j d_{j+1}, u_i d_j\})$ are:


\bea\label{b1}
&&\frac{1}{2 z k_{l,u_1}\cdot \eta}\frac{1}{2 z k_{l,u_2}\cdot \eta}\cdots \frac{1}{2 z k_{l,u_{i-2}}\cdot \eta}\frac{1}{2 z k_{d,u_{i-2}}\cdot \eta}\frac{1}{2 z k_{d,u_{i-1}}\cdot \eta}\cdots\frac{1}{2 z k_{d,u_{N-1}}\cdot \eta}\\
=&&\frac{1}{2z k_d\cdot \eta}\frac{1}{2 z k_{l,u_1}\cdot \eta}\cdots \frac{1}{2 z k_{l,u_{i-3}}\cdot \eta}(\frac{1}{2 z k_{l,u_{i-2}}\cdot \eta}-\frac{1}{2 z k_{d,u_{i-2}}\cdot \eta})\frac{1}{2 z k_{d, u_{i-1}}\cdot \eta}\cdots\frac{1}{2 z k_{d, u_{N-1}}\cdot \eta}.\nonumber
\eea
\bea\label{b2}
&&\frac{1}{2 z k_{l,u_1}\cdot \eta}\frac{1}{2 z k_{l,u_2}\cdot \eta}\cdots \frac{1}{2 z k_{l,u_{i-1}}\cdot \eta}\frac{1}{2 z k_{d,u_{i-1}}\cdot \eta}\frac{1}{2 z k_{d,u_{i}}\cdot \eta}\cdots\frac{1}{2 z k_{d, u_{N-1}}\cdot \eta}\\
=&&\frac{1}{2z k_d\cdot \eta}\frac{1}{2 z k_{l,u_1}\cdot \eta}\cdots \frac{1}{2 z k_{l,u_{i-2}}\cdot \eta}(\frac{1}{2 z k_{l,u_{i-1}}\cdot \eta}-\frac{1}{2 z k_{d,u_{i-1}}\cdot \eta})\frac{1}{2 z k_{d,u_{i}}\cdot \eta}\cdots\frac{1}{2 z k_{d,u_{N-1}}\cdot \eta}.\nonumber
\eea
\bea\label{b3}
&&\frac{1}{2 z k_{l,u_1}\cdot \eta}\frac{1}{2 z k_{l,u_2}\cdot \eta}\cdots \frac{1}{2 z k_{l,u_i}\cdot \eta}\frac{1}{2 z k_{d, u_{i}}\cdot \eta}\frac{1}{2 z k_{d, u_{i+1}}\cdot \eta}\cdots\frac{1}{2 z k_{d, u_{N-1}}\cdot \eta}\\
=&&\frac{1}{2z k_d\cdot \eta}\frac{1}{2 z k_{l,u_1}\cdot \eta}\cdots \frac{1}{2 z k_{l,u_{i-1}}\cdot \eta}(\frac{1}{2 z k_{l,u_{i}}\cdot \eta}-\frac{1}{2 z k_{d,u_{i}}\cdot \eta})\frac{1}{2 z k_{d, u_{i+1}}\cdot \eta}\cdots\frac{1}{2 z k_{d, u_{N-1}}\cdot \eta}.\nonumber
\eea
It is observed that the first term in \eref{b2} cancels the second term in \eref{b3} and the first term in \eref{b1} cancels the second term in \eref{b2}. This manner of cancellation happens for each two successive diagrams in (a) of Figure \ref{example1}, and it is found that the sum of all diagrams in (a) of Figure \ref{example1} turns out to be 0, for $\mathcal{\bar M}_1^{\mu\nu}$ calculations. When including the numerator, ie. the product of the vertexes $\mathcal{N}_{h_I}^{\mu\nu}$, the summation  of  equations such as  \eref{b1}, \eref{b2} and \eref{b3} for all the diagrams in (a) of Figure \ref{example1} is  just
\bea\label{Rec1}
\sum_{\mathcal{D}_I}{\mathcal{N}_{h_I}^{\mu\nu}\over \mathcal{C}_{h_I-1}}&\equiv&\sum_{OP\{\alpha_{u_N} \bigcup d\}}{\mathcal{N}_{h_I}^{\mu\nu}\over \mathcal{C}_{h_I-1}}\nb\\
&=&\sum_{i=1}^N{1\over 2 z k_{l,u_1}\cdot \eta}{1\over 2 z k_{l,u_2}\cdot \eta}\cdots {1\over 2 z k_{l,u_{i-1}}\cdot \eta}{1\over 2 z k_{d}\cdot \eta}{1\over 2 z k_{d,u_{i}}\cdot \eta}\cdots{1\over 2 z k_{d,u_N}\cdot \eta}\nb\\
&\times&(\mathcal{N}^{\mu\nu}_{h_I}(\cdots d, u_i, \cdots)-\mathcal{N}^{\mu\nu}_{h_I}(\cdots u_i, d, \cdots)).
\eea

For general non-adjacent BCFW shifts with $N$ up-legs and $M$ down-legs.  We can prove that the summation in \eref{barTerm} can be recombined into the summation of terms  like \eref{Rec1}.
\begin{theorem}\label{lemResum}
\bea\label{Resum}
&&\sum_{\mathcal{D}_I}{\mathcal{N}^{\mu\nu}_{h_I}(\{u_i,d_j\})\over \mathcal{C}_{h_I-1}(\{u_i,d_j\})}\equiv \sum_{OP\{\alpha_{u_N} \bigcup \alpha_{d_M}\}}{\mathcal{N}^{\mu\nu}_{h_I}(\{u_i,d_j\})\over \mathcal{C}_{h_I-1}(\{u_i,d_j\})}\nb \\
&=&\sum_{j, i=1}^{M,N}\sum_{OP\{\alpha_{u_{i-1}}\atop \bigcup \alpha_{d_{j-1}}\}}\sum_{{OP\{\alpha_{(u_{i+1},u_M)}\atop \alpha_{(d_{j+1}, d_{M})}\} }}\nb \\
&&{1\over 2 z k_{l,u_1}\cdot \eta}{1\over 2 z k_{d_1,u_1}\cdot \eta}\cdots {1\over 2 z k_{d_{j-1},u_{i-1}}\cdot \eta}{1\over 2 z (k_{d_1}+\cdots+k_{d_M})\cdot \eta}{1\over 2 z k_{d_{j},u_i}\cdot \eta}\cdots{1\over 2 z k_{d_{M-1},u_N}\cdot \eta}\nb\\ \nb\\ 
&\times&(\mathcal{N}^{\mu\nu}_{h_I}(\cdots, d_j, u_i, \cdots)-\mathcal{N}^{\mu\nu}_{h_I}(\cdots, u_i, d_j, \cdots)).\nb\\
\eea
In the last line of \eref{Resum}, only the order of nearby up-line and down-line pair, ie. $u_i$ and $d_j$ is inter-changed.  In the original form in large z limit only one of the propagators in $\mathcal{M}^{\mu\nu}(\cdots, d_j, u_i, \cdots)$ and $\mathcal{M}^{\mu\nu}(\cdots, u_i, d_j, \cdots)$ is different which is the propagator between $u_i$  and $d_j$. In the recombined summation, this different propagator is replaced with ${1\over 2 z (k_{d_1}+\cdots+k_{d_M})\cdot \eta}$ with other propagators not changed. Similar equations hold for the other summations in \eref{barTerm}. For example, for $\sum_{\mathcal{D}_1} {\mathcal{N}_{h_I-1}^{\mu\nu}\over \mathcal{C}_{h_I-1}}$, we just replace the $\mathcal{N}^{\mu\nu}_{h_I}$ in \eref{Resum} with $\mathcal{N}^{\mu\nu}_{h_I-1}$. For $\sum_{\mathcal{D}_{II}} {\mathcal{N}_{h_{II}}^{\mu\nu}\over \mathcal{C}_{h_{II}}}$, we first replace $\mathcal{N}^{\mu\nu}_{h_I}$ in \eref{Resum} with $\mathcal{N}_{h_{II}}^{\mu\nu}$. Then say the four point vertex in $\mathcal{D}_{II}$ is $\bar V_{d_j d_j+1}$, we define $d'_i=d_i$ for $i<j$, $d'_i=d_j d_{j+1}$ for $i=j$, $d'_i=d_{i+1}$ for $i>j$, $k_{d'_i}=k_{d_i}$ for $i<j$, $k_{d'_i}=k_{d_j}+k_{d_{j+1}}$ for $i=j$, $k_{d'_i}=k_{d_{i+1}}$ for $i>j$, and replace the $\{d_i\}$ in \eref{Resum} with $\{d'_i\}$. We do not repeat for other summations in \eref{barTerm}.
\end{theorem}

\begin{proof}
To prove this, we only need to prove  that each term of a fixed order of up and down type legs in the left hand side of \eref{Resum} is equal to the sum of terms in the right hand side with the same order in $\mathcal{N}^{\mu\nu}_{h_I}$. This can be done recursively. First we assume that, for each ordering of legs in $\mathcal{N}^{\mu\nu}_{h_I}$, the summation  of the right hand side of \eref{Resum} with $N-1$ up-lines and $M-1$ down-lines is  
\bea
{\mathcal{N}^{\mu\nu}_{h_I}(\cdots u_{N-1})\over\mathcal{\bar C}_{h_I-1}(\cdots u_{N-1})}
\eea
when the most right side leg is $u_{N-1}$, with
$$\frac{1}{\mathcal{\bar C}_{h_I-1}(\cdots u_{N-1})}={1\over \mathcal{C}_{h_I-1}(\cdots u_{N-1})}=
{1\over 2 z k_{u_1}\cdot \eta}{1\over 2 z (k_{u_1}+k_{d_1})\cdot \eta}\cdots  \cdots{1\over 2 z k_{d_{M-1} u_{N-2}}\cdot \eta}.
$$
Similarly for the case with the most right side leg being $d_{M-1}$, the summation is
\bea
{\mathcal{N}^{\mu\nu}(\cdots d_{M-1})\over \mathcal{\bar C}_{h_I-1}(\cdots d_{M-1})}
\eea
with ${1\over \mathcal{\bar C}_{h_I-1}(\cdots d_{M-1})}={1\over \mathcal{C}_{h_I-1}(\cdots d_{M-1})}{-2 z (k_{u_1}+\cdots+k_{u_{N-1}})\cdot \eta\over 2 z (k_{d_1}+\cdots+k_{d_{M-1}})\cdot \eta}$ and $${1\over\mathcal{C}_{h_I-1}(\cdots d_{M-1})}={1\over 2 z k_{u_1}\cdot \eta}{1\over 2 z (k_{u_1}+k_{d_1})\cdot \eta}\cdots  \cdots{1\over 2 z k_{d_{M-2} u_{N-1}}\cdot \eta}. $$
Then if we attach leg $u_N$ to the complex line following the sequence  $(\cdots u_{N-1})$, we can get 
\bea
{\mathcal{\bar C}_{h_I-1}(\cdots u_{N-1} u_{N})}={\mathcal{C}_{h_I-1}(\cdots u_{N-1} u_{N})}.
\eea

If we attach $u_N$ to the complex line following the sequence $(\cdots d_{M-1})$, we can obtain 
\bea
\mathcal{\bar C}_{h_I-1}(\cdots d_{M-1} u_{N})&=&\mathcal{\bar C}_{h_I-1}(\cdots d_{M-1}){1\over2 z (k_{u_{N-1}}+\cdots+k_{d_{M-1}})\cdot \eta}\nb\\ &&+\mathcal{C}_{h_I-1}(\cdots d_{M-1}){1\over 2 z (k_{d_1}+\cdots+k_{d_{M-1}})\cdot\eta}\nb\\
&=&\mathcal{C}_{h_I-1}(\cdots d_{M-1}u_N).
\eea
Here there is one additional contribution from changing the order of  $d_{M-1}$ and $u_N$ in the right hand side of \eref{Resum}.

Similarly, if we attach the leg $d_M$ to the complex line following the sequence $(\cdots d_{M-1})$, we can get 
\bea
{1\over\mathcal{\bar{C}}_{h_I-1}(\cdots d_{M-1} d_{M})}&=&{1\over\mathcal{\bar{C}}_{h_I-1}(\cdots d_{M-1})}{1\over2 z (k_{u_{N-1}}+\cdots+k_{d_{M-1}})\cdot \eta}\times {2 z (k_{d_1}+\cdots+k_{d_{M-1}})\cdot\eta\over 2 z (k_{d_1}+\cdots+k_{d_{M}})\cdot\eta}\nb\\ &=&{1\over \mathcal{C}_{h_I-1}(\cdots d_{M-1}d_{M})}{-2 z (k_{u_1}+\cdots+k_{u_{N-1}})\cdot \eta\over 2 z (k_{d_1}+\cdots+k_{d_{M}})\cdot \eta}.
\eea
And if attaching  the line $d_M$ to the complex line following the sequence $(\cdots u_{N-1})$,  we can get 
\bea
{1\over \mathcal{\bar{C}}_{h_I-1}(\cdots  u_{N-1} d_{M})}&=&{1\over \mathcal{\bar C}_{h_I-1}(\cdots u_{N-1})}{2 z (k_{d_1}+\cdots+k_{d_{M-1}})\cdot\eta\over 2 z (k_{d_1}+\cdots+k_{d_{M}})\cdot\eta}{1\over2 z (k_{u_{N-1}}+\cdots+k_{d_{M-1}})\cdot \eta}\nb\\ &&+{1\over \mathcal{C}_{h_I-1}(\cdots u_{N-1})}{-1\over 2 z (k_{d_1}+\cdots+k_{d_{M}})\cdot\eta}\nb\\
&=&{1\over \mathcal{\bar C}_{h_I-1}(\cdots u_{N-1}d_M)}{-2 z (k_{u_1}+\cdots+k_{u_{N-1}})\cdot \eta\over 2 z (k_{d_1}+\cdots+k_{d_{M}})\cdot \eta}.
\eea

Thus for $N$ up legs and $M$ down legs, we get:
\bea
\mathcal{\bar C}_{h_I-1}(\cdots u_{N})&=&\mathcal{C}_{h_I-1}(\cdots u_{N}) \nb\\
{1\over  \mathcal{\bar{C}}_{h_I-1}(\cdots  d_{M})}&=&{1\over \mathcal{C}_{h_I-1}(\cdots d_M)}{-2 z (k_{u_1}+\cdots+k_{u_{N}})\cdot \eta\over 2 z (k_{d_1}+\cdots+k_{d_{M}})\cdot \eta}.
\eea
With momenta conservation and the shift condition \eref{conditionEta} it is easy to see 
\bea
{-2 z (k_{u_1}+\cdots+k_{u_{N}})\cdot \eta\over 2 z (k_{d_1}+\cdots+k_{d_{M}})\cdot \eta}=1.
\eea
By induction, the equation \eref{Resum}, ie. \tref{lemResum}, has been proved. \endofproof
\end{proof}
\begin{corollary}\label{corResum}
When the $\mathcal{N}^{\mu\nu}_{h_I}(\{u_i,d_j\})$ are independent of the relative orders of the external legs, we have 
\be
\sum_{\mathcal{D}_I}{1\over \mathcal{C}_{h_I-1}}=0.
\ee
Such equations hold also for the other cases in \eref{barTerm}. For example, $\sum_{\mathcal{D}_{II}}{1\over \mathcal{C}_{h_{II}}}=0$ and $\sum_{\mathcal{D}_{III}}{1\over \mathcal{C}_{h_{III}+1}}=0$.
\end{corollary}

\subsection{Amplitudes in the Large $z$ Limit under Non-adjacent BCFW Shifts}
\subsubsection{$\mathcal{O}(z^1)$ Behavior of the Amplitudes}

To obtain the $\mathcal{O}(z^1)$ behavior of the amplitude $\mathcal{M}^{\mu\nu}$, we only need the case 1 in \sref{theorem1Application}, that is $\mathcal{D}_I$ with all the reduced three point vertexes taking their S term components. Furthermore we only need to keep the terms with highest order of z in all the vertexes and propagators, ie. $\mathcal{N}_{h_I}^{\mu\nu}$ and $\mathcal{C}^{h_I-1}$. The $z$ order of some S term $S_{u/d}$ does not depend on its position on the the complex line. As a result, $\mathcal{N}_{h_I}^{\mu\nu}\propto g^{\mu\nu}$ are the same for all diagrams of type $\mathcal{D}_I$ and we obtain:
\beq
\mathcal{M}_1^{\mu\nu}=\sum_{\mathcal{D}_I}\frac{\mathcal{N}_{h_I}^{\mu\nu}}{\mathcal{C}^{h_I-1}}=\mathcal{N}_{h_I}^{\mu\nu}\sum_{\mathcal{D}_I}\frac{1}{\mathcal{C}^{h_I-1}}=0.
\label{orderz}
\eeq
The second equation is from \cref{corResum}. The external lines can be either off-shell or on-shell. In conclusion, $\mathcal{O}(z^1)$ of $\mathcal{M}^{\mu\nu}$ for non-adjacent shifts vanish.

\subsubsection{$\mathcal{O}(z^0)$ Behavior of the Amplitudes}
In this subsection, we are going to show that: for non-adjacent shifts, 
\be\label{Mz0}
\mathcal{M}_0^{\mu\nu}\propto g^{\mu\nu}.
\ee

Using \eref{LargeExap} and \eref{barTerm}, we can classify the terms that contribute to $\mathcal{M}_0^{\mu\nu}$ into the following groups:
\begin{itemize}
\item $\sum_{\mathcal{D}_I} {c_{h_I-2} \mathcal{N}_{h_I}^{\mu\nu}\over c_{h_I-1}^2 }\propto g^{\mu\nu}$, since $\mathcal{N}_{h_I}^{\mu\nu}$ is proportional to $g^{\mu\nu}$ in diagrams $\mathcal{D}_I$. 
\item $\sum_{\mathcal{D}_{II'}} {\mathcal{N}_{h_{II'}}^{\mu\nu}\over c_{h_{II'}}}\propto g^{\mu\nu}$.  In $\mathcal{D}_{II'}$,  there is one reduced four point vertex ${\bar V}_{u_i d_j}$ in the complex line. And all the others are reduced three point vertexes with only their $S$ term components. According to the forms of ${\bar V}_{u_i d_j}$ and $S$ term, it is easy to see $\mathcal{N}_{h_{II'}}^{\mu\nu}\propto g^{\mu\nu}.$
\item $\sum_{\mathcal{D}_{II}} {(\mathcal{N}_{II})_{h_{II}}^{\mu\nu}\over (c_{II})_{h_{II}}}=0$, using \cref{corResum}, essentially the same as in \eref{orderz}.
\item $\sum_{\mathcal{D}_I} {\mathcal{N}_{h_I-1}^{\mu\nu}\over c_{h_I-1}}\propto g^{\mu\nu}.$ $\mathcal{D}_I$ are the diagrams comprised all of reduced three point vertexes. There are two contributions to this summation. One contribution is when only one of the reduced three point vertexes takes its R term part and other vertexes take their S components. Without loss of generality, we assume the vertex with the leg $u_i$ takes its R part. All these diagrams have the same $\mathcal{N}_{h_I-1}$. According to \cref{corResum}, the sum of all these diagrams contribute 0 to $\sum_{\mathcal{D}_I} {\mathcal{N}_{h_I-1}^{\mu\nu}\over c_{h_I-1}}$. The other contribution is when all the reduced three point vertexes take their $S$ term components. This contribution is obviously proportional to $g^{\mu\nu}$.
\end{itemize}
Thus we have proven that for non-adjacent shifts, $\mathcal{M}_0^{\mu\nu}$ is proportional to $g^{\mu\nu}$.

\subsubsection{$\mathcal{O}(z^{-1})$ Behavior of the Amplitudes}\label{order1overz}
The previous two sub sections do not depend on whether the external legs are on-shell or off-shell. In this sub section, we discuss $\mathcal{M}_{-1}^{\mu\nu}$ in the two cases when the external lines are all on-shell and when some of them are off-shell.

When all external lines are on shell, the "generalized Ward identities" in \lref{GWI1} and \lref{GWI2} become the real Ward identities where the expressions are exactly zero. Thus the last term for $\mathcal{M}_{-1}^{\mu\nu}$ in \eref{LargeExap}, ie. $\mathcal{M}_{-1(M)}^{\mu\nu}$, is 0. By the similar arguments as in the last sub section,  it is easy to see that each other term except $\mathcal{\bar M}_{-1}^{\mu\nu}$ in the third equation of \eref{LargeExap} is in the form of $A g^{\mu\nu}+B^{\mu\nu}$ with $B^{\mu\nu}$ antisymmetric in $\mu$ and $\nu$. We are going to concentrate on terms that contribute to $\mathcal{\bar M}_{-1}^{\mu\nu}$ in \eref{barTerm}:
\begin{itemize}
\item $\sum_{\mathcal{D}_I} {\mathcal{N}_{h_I-2}^{\mu\nu}\over c_{h_I-1}}\propto A g^{\mu\nu}+B^{\mu\nu}$. In $\mathcal{D}_I$, all the vertexes in the complex line are the three point vertexes $\bar V_{u/d}$. We can classify them into the following groups:  \\ \textbf{\textcircled{\bf{a}}} When $\bar V_{u/d}$ all take their $S$-term components or only one of them takes its $R$ term part, such contributions are obviously of form $A g^{\mu\nu}+B^{\mu\nu}$. \\ \textbf{\textcircled{\bf{b}}} When  the two vertexes with R parts are all above (or below) the complex line, for example $R_{u_i}$ and $R_{u_j}$, and others taking S terms, $\mathcal{N}_{h_I-2}^{\mu\nu}$ are the same for all these diagrams. Thus, same to \eref{orderz}, using \cref{corResum}, these terms contribute 0 to $\mathcal{\bar M}_{-1}^{\mu\nu}$. \\ \textbf{\textcircled{\bf{c}}} When the two vertexes with R parts are $R_{u_i}$ and $R_{d_j}$, with indices $\mu_{u_i}$ and $\mu_{d_j}$, other vertexes are all taking S components. Furthermore since each R term decreases order of z by 1 compared to S term, to contribute to the next to next order of the product of the vertexes ie. $\mathcal{N}_{h_I-2}$, each S term of other vertexes contributes the same to $\mathcal{N}_{h_I-2}$ regardless of its position on the complex line. $R_{u_i}$ and $R_{d_j}$ are also independent of their positions on the complex line. Thus as for the calculation of $\mathcal{N}_{h_I-2}$, we can regard $S_{u_i'}$ and $S_{d_j'}$ as commuting, $S_{u_i'}$ and $R_{d_j}$ commuting, and $R_{u_i}$ and $S_{d_j'}$ commuting. Applying \tref{lemResum}, we can see that the only non-vanishing terms are from:
\begin{small}
\begin{equation*}
\mathcal{N}_{h_I-2}^{\mu\nu}(\cdots, d_j, u_i, \cdots)-\mathcal{N}_{h_I-2}^{\mu\nu}(\cdots, u_i, d_j, \cdots)\propto (R_{d_j})^\mu{}_{\rho}{}^{\mu_{d_j}}(R_{u_i})^{\nu\rho}{}^{\mu_{u_i}}-(R_{u_i})^{\rho\mu}{}^{\mu_{u_i}}(R_{d_j})_\rho{}^\nu{}^{\mu_{d_j}},
\end{equation*}
\end{small}

which is antisymmetric in $\mu$ and $\nu$, invoking that $R$ term is antisymmetric in its first two indices, referring to \eref{newV3s}.


\item $\sum_{\mathcal{D}_{II}} {\mathcal{N}_{h_{II}-1}^{\mu\nu}\over c_{h_{II}}}\propto A g^{\mu\nu}+B^{\mu\nu}$. In $\sum_{\mathcal{D}_{II}}$, the diagrams are comprised of one reduced four point vertex, which is not ${\bar V}_{u_i d_j}$, and the rest vertexes are reduced three point vertexes, one of which takes its R term part. In the definition of the reduced vertexes \eref{ReduV}, we call the last term of $\bar{V}_{u_i u_j}$ or $\bar{V}_{d_i d_j}$ as symmetric term and the first two terms as antisymmetric term. The discussion is parallel to the case above: \\ \textbf{\textcircled{\bf{a}}}   Only one or none of the reduced four point vertex and the R term takes its anti-symmetric part. The contribution is of form $A g^{\mu\nu}+B^{\mu\nu}$. \\ \textbf{\textcircled{\bf{b}}} The vertex with R term and the four point vertex are both above (or below) the complex line. It contributes 0 to $\sum_{\mathcal{D}_{II}} {\mathcal{N}_{h_{II}-1}^{\mu\nu}\over c_{h_{II}}}$.  \\ \textbf{\textcircled{\bf{c}}} The vertex with R term and the four point vertex are on the opposite sides of the complex line. The contribution is antisymmetric in $\mu$ and $\nu$.

\item $\sum_{\mathcal{D}_{III}} {\mathcal{N}_{h_{III}}^{\mu\nu}\over c_{h_{III}+1}}\propto A g^{\mu\nu}+B^{\mu\nu}. $ In $\mathcal{D}_{III}$, the diagrams are comprised of two reduced four point vertexes, neither of which is ${\bar V}_{u_i d_j}$, and the other reduced three point vertexes all take their S term parts. The discussion is again parallel to the cases above: \\ \textbf{\textcircled{\bf{a}}}   Only one or none of the reduced four point vertexes takes its anti-symmetric part. The contribution is of form $A g^{\mu\nu}+B^{\mu\nu}$. \\ \textbf{\textcircled{\bf{b}}} The two reduced four point vertexes both take their anti-symmetric parts and are both above (or below) the complex line. It contributes 0 to $\sum_{\mathcal{D}_{III}} {\mathcal{N}_{h_{III}}^{\mu\nu}\over c_{h_{III}+1}}$. \\ \textbf{\textcircled{\bf{c}}} The two reduced four point vertexes take their antisymmetric parts and are on the opposite sides of the complex line. The contribution is antisymmetric in $\mu$ and $\nu$.

\end{itemize}

Above all, when all the external legs are on shell, for non-adjacent shifts, $\mathcal{O}(z^{-1})$ of $\mathcal{M}^{\mu\nu}$, ie. $\mathcal{M}_{-1}^{\mu\nu}$, is in form of a metric term plus a term antisymmetric in $\mu$ and $\nu$.

Now we discuss the case when some external lines are off-shell. The additional contribution is from the last term $\mathcal{M}_{-1(M)}^{\mu\nu}$ in \eref{LargeExap}, which is from the diagrams (a) and (b) of Figure \ref{Mtermcancel}. We analyze how the diagrams contribute to $\mathcal{M}_{-1}^{\mu\nu}$. Take the diagram (a) for example, with the last $M^R$ factor to be $M_{d_i}^R$ (same analysis for $M_{u_{i}}^R$). Assume the next vertex is $V_{u_{j}}$ (same analysis for $V_{d_{j}}$). Then $M_{d_i}^R\  V_{u_{j}}$ can be decomposed according to \eref{kdotV} and Figure \ref{vertexnotation}, see Figure \ref{MRdotV}.

\begin{figure}[htb]
\centering
\includegraphics{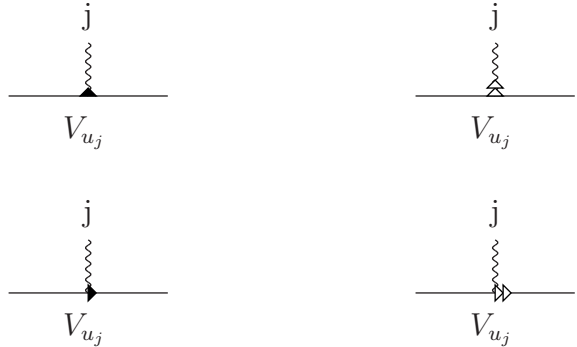}
\caption{Decomposition of $M_{d_i}^R\  V_{u_j}$ in the notations of Figure \ref{vertexnotation}. The horizontal line is the complex line, and photon line represents external leg.}
\label{MRdotV}
\end{figure}

Among the four terms in Figure \ref{MRdotV}, the first line two terms combined is in the form
\beq
k_j^2 g^{\mu_j\delta}-k_j^{\mu_j} k_j^\delta,
\label{kkterm}
\eeq
where $\delta$ is some index we do not care here. The first term in the second line of Figure \ref{MRdotV} need not be considered since they will cancel in group in the manner of Figure \ref{treecancel2}. In this cancellation, diagrams with some vertexes outside the complex line is involved, but it does not affect the property of our conclusion, once we apply less point results to these diagrams. The second term in the second line of Figure \ref{MRdotV} acts on the next vertex on the complex line, and can be analyzed in the same steps as in this paragraph. Only when the vertex being acted on is the last vertex on the complex line, the second line two terms of Figure \ref{MRdotV} should be retained, which sum up to equal $k_r^2 g^{\nu\delta}-k_r^{\nu} k_r^\delta$, also in the form of \eref{kkterm}. (b) of Figure \ref{Mtermcancel} is similarly analyzed, and results in terms in the form of \eref{kkterm}. \eref{kkterm} is 0 when $k_j$ is on shell and only receives contributions from off shell external legs. Thus we can make the conclusion that the additional contribution to $\mathcal{M}_{-1}^{\mu\nu}$ from off shell external legs is:
\beq
\sum_{\mbox{\tiny off shell\ }j} (k_j^2 g^{\mu_j\delta}-k_j^{\mu_j} k_j^\delta)\cdots,
\label{kkform}
\eeq
where the sum is over each off shell external leg.

Direct calculation shows that \eref{kkform} is antisymmetric in $\mu$ and $\nu$ when there is only 1 leg above and 1 leg below the complex line, and not antisymmetric for 5 point amplitudes, unlike to be antisymmetric for more point amplitudes.

In conclusion, for non adjacent BCFW shifts of on shell tree amplitudes, $\mathcal{O}(z^{-1})$ of $\mathcal{M}^{\mu\nu}$ is in form of a metric term plus a term antisymmetric in $\mu$ and $\nu$; for amplitudes with off shell legs, $\mathcal{O}(z^{-1})$ has additional contributions from the off shell legs in the form of \eref{kkform}, which manifestly vanishes when the legs become on shell. We guess that for on shell loop level amplitudes, terms in \eref{kkform} may cancel the contribution from ghost loops, which deserves further investigation.

\section{Conclusion}

In this article, we have carefully analyzed the boundary behaviors of pure Yang-Mills amplitudes under adjacent and non adjacent BCFW shifts in Feynman gauge. We introduced reduced vertexes for Yang-Mills fields, proved that these reduced vertexes are equivalent to the original vertexes, as for the study of boundary behaviors, which greatly simplifies our analysis of boundary behaviors. Boundary behaviors for adjacent shifts are readily obtained using reduced vertexes. Then we find that the boundary behaviors for non-adjacent shifts are much better than those of adjacent shifts. Comparing to adjacent shifts, non adjacent shifts allow us to permute the external legs while retaining color ordering. We proved a theorem about permutation sum, which plays key roles in our analysis of non-adjacent boundary behaviors besides the use of reduced vertexes, and the theorem is the essential reason for the improvement of boundary behaviors for non adjacent shifts compared to adjacent shifts. The conclusions are, $\mathcal{O}(z^{1})$ of $\mathcal{M}^{\mu\nu}$ is proportional to metric $g^{\mu\nu}$ for adjacent shifts, and vanishes for non adjacent shifts; $\mathcal{O}(z^{0})$ of $\mathcal{M}^{\mu\nu}$ is metric term plus antisymmetric term for adjacent shifts, and is proportional to $g^{\mu\nu}$ for non adjacent shifts. Based on the boundary behaviors, we find that it is possible to generalize BCFW recursion relation to calculate general tree level off shell amplitudes, with the aid of our previous papers \cite{Chen1,Chen2,Chen3}. The procedure is described in the second section, before we discuss boundary behaviors.

We proved that boundary behaviors at $\mathcal{O}(z^1)$ and $\mathcal{O}(z^0)$ do not depend on whether the external legs are on shell or not. We also analyzed the $\mathcal{O}(z^{-1})$ behavior for non adjacent shifts. When all the external legs are on shell, $\mathcal{O}(z^{-1})$ of $\mathcal{M}^{\mu\nu}$ is metric term plus antisymmetric term. When some external legs are off shell, we also give the general form of the contribution to $\mathcal{O}(z^{-1})$ from each off shell leg, which manifestly vanishes when the leg becomes on shell. For on shell loop level amplitudes, the loop lines can be dealt with as off shell legs here and has the contribution to $\mathcal{O}(z^{-1})$ in the form we have obtained, which seems very likely to cancel the ghost loop contributions, resulting in some good $\mathcal{O}(z^{-1})$ behaviors for loop level non adjacently shifted on shell amplitudes. This deserves our further investigation.

Our conclusions on boundary behaviors in Feynman gauge are consistent with those in AHK gauge in \cite{Boels,Nima1}. Our work has two major advantages. First, the necessary conditions are given explicitly in our discussion on the boundary behaviors. According to this,  we can present a procedure to calculate general tree level off shell amplitudes using BCFW technique and the technique in \cite{Chen3}. And the second is related to our permutation sum theorem, ie. \tref{lemResum}. This theorem tells us why the amplitudes with non-adjacent BCFW shifts have improved boundary behaviors. Actually, in \cite{Boels} there are several important assumptions about the relationship between the improved boundary behaviors and the general permutation sums. Hopefully, some generalization of our theorem here will be helpful for the proof of these assumptions. This will be left for further work.


{\textbf{Acknowledgement} We thank Yijian Du for helpful discussions. This work is funded by the Priority Academic Program Development of Jiangsu Higher Education Institutions (PAPD), NSFC grant No.~10775067, Research Links Programme of Swedish Research Council under contract No.~348-2008-6049, the Chinese Central Government's 985 Project grants for Nanjing University, the China Science Postdoc grant no. 020400383. the postdoc grants of Nanjing University 0201003020}

%

\end{document}